\documentclass[12pt]{article}
\usepackage[margin=1.5in]{geometry}
\usepackage{amsmath}
\usepackage{amssymb}
\usepackage{amsthm}
\usepackage{amssymb}
\usepackage{amsmath}
\usepackage{graphicx}
\usepackage{url}

\newtheorem{theorem}{Theorem}

\newtheorem{lemma}{Lemma}
\newtheorem{corollary}{Corollary}
\newtheorem{definition}{Definition} 

\newtheorem{remark}{Remark}

\newcommand{\HC}{\mathcal{H}}
\newcommand{\G}{\mathcal{G}}

\newcommand{\E}{\mathbb{E}}
\newcommand{\FC}{{F}}
\newcommand{\PC}{{P}}
\newcommand{\NC}{{N}}

\newcommand\blfootnote[1]{%
  \begingroup
  \renewcommand\thefootnote{}\footnote{#1}%
  \addtocounter{footnote}{-1}%
  \endgroup
}

\begin{document}
\title{Query Complexity of Correlated Equilibrium}
\author{Yakov Babichenko\footnote{Center for the Mathematics of Information, Department of Computing and Mathematical Sciences, California Institute of Technology. E-mail: {\tt babich@caltech.edu.} The author  gratefully acknowledges support from a Walter S. Baer and Jeri
Weiss fellowship. } \and Siddharth Barman\footnote{Center for the Mathematics of Information, Department of Computing and Mathematical Sciences, California Institute of Technology. E-mail: \tt{barman@caltech.edu}}}
\date{}
\maketitle

\begin{abstract}
We study lower bounds on the query complexity of determining correlated equilibrium. In particular, we consider a query model in which an $n$-player  game is specified via a black box that returns players' utilities at pure action profiles. In this model we establish that in order to compute a correlated equilibrium any \emph{deterministic} algorithm must query the black box an exponential (in $n$) number of times. \blfootnote{After the completion of this result we became aware of  a recent and independent work by Hart and Nisan~\cite{HN} that generalizes the result presented in this paper. In particular,  Hart and Nisan~\cite{HN} establish query complexity lower bounds for  randomized algorithms and computing approximate equilibria; for a discussion of this work see Section~\ref{sec:dis-hn}.}

\end{abstract}

\section{Introduction}
Equilibria are fundamental constructs in game theory that formally specify potential outcomes in strategic settings. Nash equilibrium~\cite{Nash} and its generalization, correlated equilibrium~\cite{aumann74,aumann87}, are arguably the two most well-established examples of such notions of rationality. Both of these concepts denote distributions over strategies of players at which no player could benefit by unilateral deviation. What distinguishes these constructs is the fact that mixed Nash equilibrium is defined to be a product of independent distributions (one of every player), whereas a correlated equilibrium is a general (joint) probability distribution over the strategy space. 

Questions related to the complexity of determining equilibria have been a driving force behind research at the intersection of computer science and economics. Recent results imply that it is unlikely that there exists an efficient algorithm for  determining mixed Nash equilibrium, even in the two-player case (see~\cite{goldberg,costis,chen}). On the other hand, given a \emph{fixed} number of players $n$ and and at most $m$ actions for every player, a correlated equilibrium can be computed in time polynomial in $m$; since, correlated equilibria can be specified via a linear feasibility program with $O(nm^2)$ linear inequalities over $O(m^n)$ variables. But, this leaves open a natural question of whether there exists an efficient algorithm that determines a correlated equilibrium even if the number of players is large. In this paper we address this question. 

Note that if the game is specified in normal form (i.e., the utilities of all the players at every action profile in the game is given to the algorithm) then the input size itself is exponential in $n$; but, our objective is to determine if a correlated equilibrium can be computed in time polynomial in $n$. A standard approach to deal with such scenarios is to assume that the game is specified via a \emph{black box}. The algorithm may \emph{query} the black box about the game, and receive answers in $O(1)$ time (or in $poly(n,m)$ time). 

The core modeling question is which type of queries can the black box answer? An probable model is one in which every query is a product distribution over strategies, $x=(x_1,x_2,...,x_n)$ (i.e., a mixed action profile), and the black box returns the expected value of the utilities of the players,  $\left(\E_{s \sim x} [ u_i(s)]\right)_i$, as an answer. Such a model is applicable if the game has a \emph{succinct representation} (see \cite{JL}); since, in such a case the expected utility of each player can be computed in $poly(n,m)$ time for every mixed action profile. For such a querying model, building upon the work of Papadimitriou and Roughgarden \cite{PR}, Jiang and Leyton-Brown \cite{JL} proved that there exists a polynomial-time algorithm for computing exact correlated equilibrium. In particular, polynomial (in $n$ and $m$) number of queries are required in this model.

But, for general games it is not always reasonable to assume that such a black box exists. The support of a mixed action profile may be exponential in $n$; therefore, the existence of a black box, which aggregates over those exponential number of outcomes, is a strong assumption. 

A more applicative model is one in which the queries are \emph{pure-action} profiles (i.e., queries are of the form $(s_1, s_2, \ldots, s_n)$, where $s_i$ is a specific strategy of player $i$) and the black box returns the utilities of players at those profiles.  For example, in a repeated-game framework~\cite{hart2000simple} where players do not know the game, they observe the outcome of the realized pure action even if they played a mixed strategy. It might seem that the model is too weak for efficiently computing a correlated equilibrium, because only a very small fraction of the game is known after a polynomial number of queries. But, surprisingly, it turns out that an \emph{approximate} correlated equilibrium\footnote{A probability distribution $\sigma$ over the players' strategies is said to be an approximate ($\epsilon$) correlated equilibrium if for any player unilaterally deviating from strategies drawn from $\sigma$ increases utility, in expectation, by at most $\epsilon$.} can be computed using such a black box in polynomial time. This can be done using procedures--in particular, regret-minimizing dynamic--developed in \cite{HM}, \cite{FV}, and \cite{LW}. These dynamics converge to an {approximate} correlated equilibrium in polynomial number of steps.

Another observation that brings up the applicability of the pure-action-query model is as follows: in every game there exists a correlated equilibria with polynomial-sized support (e.g.,~\cite{germano2007existence}). This is because correlated equilibria are defined by a polynomial  (specifically, $O(nm^2)$) number of linear inequalities, therefore, a basic feasible solution of such a linear feasibility program will have a polynomial number of non-zero entries.  In addition, using only pure-action queries we can efficiently verify whether a probability distribution with polynomial-sized support is a correlated equilibrium or not. 
 
Hence, in this context, it is natural to ask if an \emph{exact} correlated equilibrium can be computed in polynomial time using pure-action queries. This was an open question posed by Sergiu Hart~\cite{open}. 

In this paper we answer this question in the negative: the number of pure-action queries that are required to find an \emph{exact} correlated equilibrium is exponential in $n$, even for games with two actions per player. This result shows that the algorithms in~\cite{PR} and~\cite{JL} that use mixed-action queries (i.e., a succinct representation of the game) are the best possible, in the sense that if we can evaluate utilities only at pure-action profiles and not at mixed-action profiles then a deterministic polynomial-time algorithm that computes a correlated equilibrium does not exist.

We note that similar black-box models have been previously considered by Hirsch et al.~\cite{HPV} and Fernley et al.~\cite{GS} in the context of determining fixed points and Nash equilibrium respectively. 

\section{Notation and Preliminaries}

We consider games with $n$ players, two actions $\{0,1\}$ per
player, and a utility function $u_i: \{0,1\}^n \rightarrow \mathbb{R}$ for every player $i \in [n]$.
We will use $V:=\{0,1\}^n$ to denote the set of pure actions in the
game, and also the set of vertices in the $n$-dimensional
hypercube $\HC_n$. For $s \in V$, we denote by $s\lnot
i:=(s_1,...,s_{i-1},1-s_i,s_{i+1},...,s_n)$ the neighbor of $s$ in
the hypercube obtained by switching the $i$th coordinate.
Write $d_i(s)=u_i(s)-u_i(s\lnot i)$ for the difference \footnote{Note that $d_i(s)$ is negative of the standard regret. We use it instead of regret for ease of exposition.} in the utilities of player $i$ if she switches her action. Finally, let $d(s):=(d_i(s))_{i=1}^n$ 
denote the vector of differences, and $D(s):=\sum_i d_i(s)$ be
the sum of those differences.

In a binary-action game, a correlated equilibrium is defined as follows. 

\begin{definition} \label{defn:ce}
A (joint) probability distribution $\sigma$ over $\{0,1\}^n$ is a correlated equilibrium if for every player $i$ and action $s_i  \in \{0,1\}$ we have
\begin{align*}
\sum_{s_{-i}} \left[ u_i(s_i, s_{-i}) - u_i(1-s_i, s_{-i}) \right] \sigma(s_i, s_{-i} ) \geq 0,
\end{align*}
where $s_{-i}$ denotes the strategies chosen by players other than $i$.   
\end{definition}

We write $s=(s_i,s_{-i})$ and use the definition of $d_i(s)$ to get \[ \E_{s \sim \sigma} [ d_i(s) ] = \sum_{s_i} \sum_{s_{-i}} \left[ u_i(s_i, s_{-i}) - u_i(1-s_i, s_{-i}) \right] \sigma(s_i, s_{-i} ).\] Therefore, if $\sigma$ is a correlated equilibrium then the following holds for every player $i$ 
\begin{equation}\label{eq-ce}
\E_{s \sim \sigma} [d_i(s)]\geq 0.
\end{equation}
Note that inequality (\ref{eq-ce}) is a necessary condition for $\sigma$ to be a correlated equilibrium, it is not sufficient.

\begin{remark}\label{rem}
Given distribution $\sigma$, if $\ \E_{s \sim \sigma}[ D(s)]<0$ then there exists player $i$ such
that $\E_{s \sim \sigma}[d_i(s)]<0$. Hence, any distribution that satisfies 
$\ \E_{s \sim \sigma}[ D(s)]<0$ cannot be a correlated equilibrium.  
\end{remark}

We will also consider \emph{coarse correlated equilibrium} meaning probability distributions $\pi$ over strategy profiles that satisfy 
\begin{align*}
\E_{ s \sim \pi } [u_i(s) ] \geq \E_{s \sim \pi } [u_i(s_i', s_{-i})],
\end{align*}
for every player $i$ and strategy profile $s_i'$ in $i$'s action set (which in general can contain more than two actions). A coarse correlated equilibrium is a generalization of correlated equilibrium in which a player's deviation ($s_i'$) is committed to in advance and independent of the sampled strategy profile ($s$). Coarse correlated equilibria are sometimes called the Hannan set, e.g., see~\cite{Young}.

\begin{remark}\label{rem2}
In binary-action games the set of correlated equilibria
coincides with the set of coarse correlated equilibria.
\end{remark}

To prove our result we will need the following giant-component lemma over hypercubes. For subsets $S, T \subseteq V$ we denote by $\delta(S,T)$ the set of
edges in the hypercube $\HC_n$ that connect $S$ and $T$.

\begin{lemma}
\label{lemma:giant} For any subset of vertices $S \subset V$ of
cardinality less than $\frac{2^n}{n^2+1}$, the number of vertices in
the largest connected component of $V \setminus S$ is greater than
$2^{n-1}$.
\end{lemma}

\begin{proof}
We denote the edge expansion of the hypercube by $h(\HC_n)$. That
is,
\begin{align*}
h(\HC_n) := \min_{T \subset V} \  \frac{|\delta(T, V \setminus
T)|}{\min \{|T|, |V \setminus T |\} }
\end{align*}

Since the second-largest eigenvalue of the hypercube is $1 -
\frac{2}{n}$ (see, e.g., \cite{Tr}), using Cheeger's inequality (see, e.g., \cite{Ch}) we get that $h(\HC_n)\geq 1/n$. Therefore,  for any $C \subset
\mathcal{V}$ we have
\begin{align}
\label{ineq:exp} |\delta(C, V \setminus C)| \geq \frac{1}{n} \min
\{|C|, |\mathcal{V} \setminus C |\}.
\end{align}

Say we are given an $S \subset V$ that satisfies  $|S| <
\frac{2^n}{n^2+1}$. Denote by $C_1,...,C_k$ the connected component
of $V \setminus S$. Assume for contradiction that all the connected
components satisfy $|C_k|\leq 2^{n-1}$, then $|C_k|=\min \{|C_k|, |V
\setminus C_k |\}$. Since all the edges from $C_k$ to $V\setminus C_k$
must end up in $S$ (and not in other connected components), there
are at least

\begin{equation}
\sum_k \frac{|C_k|}{n} > \frac{1}{n}2^n \left(1-\frac{1}{n^2+1} \right)=
\frac{1}{n}2^n \left(\frac{n^2}{n^2+1}\right)
\end{equation}
incoming edges into $S$, which implies that there are at least
$\frac{1}{n^2}2^n (\frac{n^2}{n^2+1})$ vertices in $S$ (because
every vertex has at most $n$ incoming edges), which contradicts the
assumption on the size of $S$.

\end{proof}

\section{Lower Bound}

Our main result is as follows:

\begin{theorem}\label{theo:main}
\label{theorem:main} Let $A$ be a deterministic algorithm that, for
any $n$-player binary-action game $\G$, determines a correlated
equilibrium after asking $q$ pure-action queries. Then $q \geq
\frac{2^n}{n^2+1}$.
\end{theorem}

This is in contrast to mixed-action queries where $q=poly(n,m)$ 
suffices~\cite{PR,JL} . This is also in contrast to probabilistic regret-minimizing
dynamics (e.g., regret matching \cite{HM}) that determine correlated $\varepsilon$-equilibrium using $ poly(n,m,\frac{1}{\varepsilon})$ pure
action queries. 

We get the following corollary from Remark~\ref{rem2}. 
\begin{corollary}
Let $A$ be a deterministic algorithm that, for any $n$-player
binary-action game $\G$, determines a coarse correlated equilibrium
after asking $q$ pure-action queries. Then $q \geq
\frac{2^n}{n^2+1}$.
\end{corollary}

\begin{proof}[Proof of Thoerem \ref{theo:main}]

The high-level argument behind the result is that if we are given
subexponentially many strategy profiles (i.e., vertices of the
hypercube), one after the other, then we can \emph{adaptively}
set utilities at these profiles such that no correlated
equilibrium can be constructed using them. To develop some intuition
as to how we accomplish this construction, say $Q$ is the set of queried (pure) action profiles and we set the utilities
such that the sum of difference values, $D(s)$, is equal to $-1$ for all $s \in Q$. Then by
Remark~\ref{rem}, there does not exist a correlated equilibrium
$\sigma$ whose support is contained in $Q$. The querying algorithm
might use strategy profiles that are not in $Q$, but the fact that
the utilities at these strategies is unknown to the algorithm lets us
ensure that any proposed distribution is not a correlated
equilibrium.

As stated above, a key element for us is an adaptive process that
given a sequence of strategy profiles sets utilities such that
 $D(s)$ is negative at every queried strategy. We model
this as a multi-round interaction between a querier (surrogate for
the deterministic algorithm) and an adversary (representing the black
box). In each round, the querier sends in a strategy profile $s$ to the
adversary which in turn returns a vector of differences $d(s)$. Note
that in a binary-action game the vector $d(s)$ must satisfy the
constraint
\begin{equation}\label{eq-cons}
d_i(s\lnot i)= -d_i(s)
\end{equation}
for every player $i$. And those are the \emph{only} constraints on
the difference vectors $d(s)$ in the game. Hence the adversary
needs to ensure that the returned vectors satisfy the
complementarity constraint (\ref{eq-cons}).

It is worth mentioning that unlike the black box the adversary is
designed to return difference values, $d(s)$. But, this does not reduce the
applicability of the interaction model. Given that the
adversary returns difference values that satisfy property $(\ref{eq-cons})$ for
queried strategy profiles $s$, we can set utilities, at $s$ and
its neighbors (specifically $u_i$ at $s \lnot i$ for all $i$), that match the difference values
reported by the adversary. These utilities can then be considered as
the response of the black box at $s$.

\begin{figure}[h]\label{fig:proof}
\begin{center}
\includegraphics[scale=0.6]{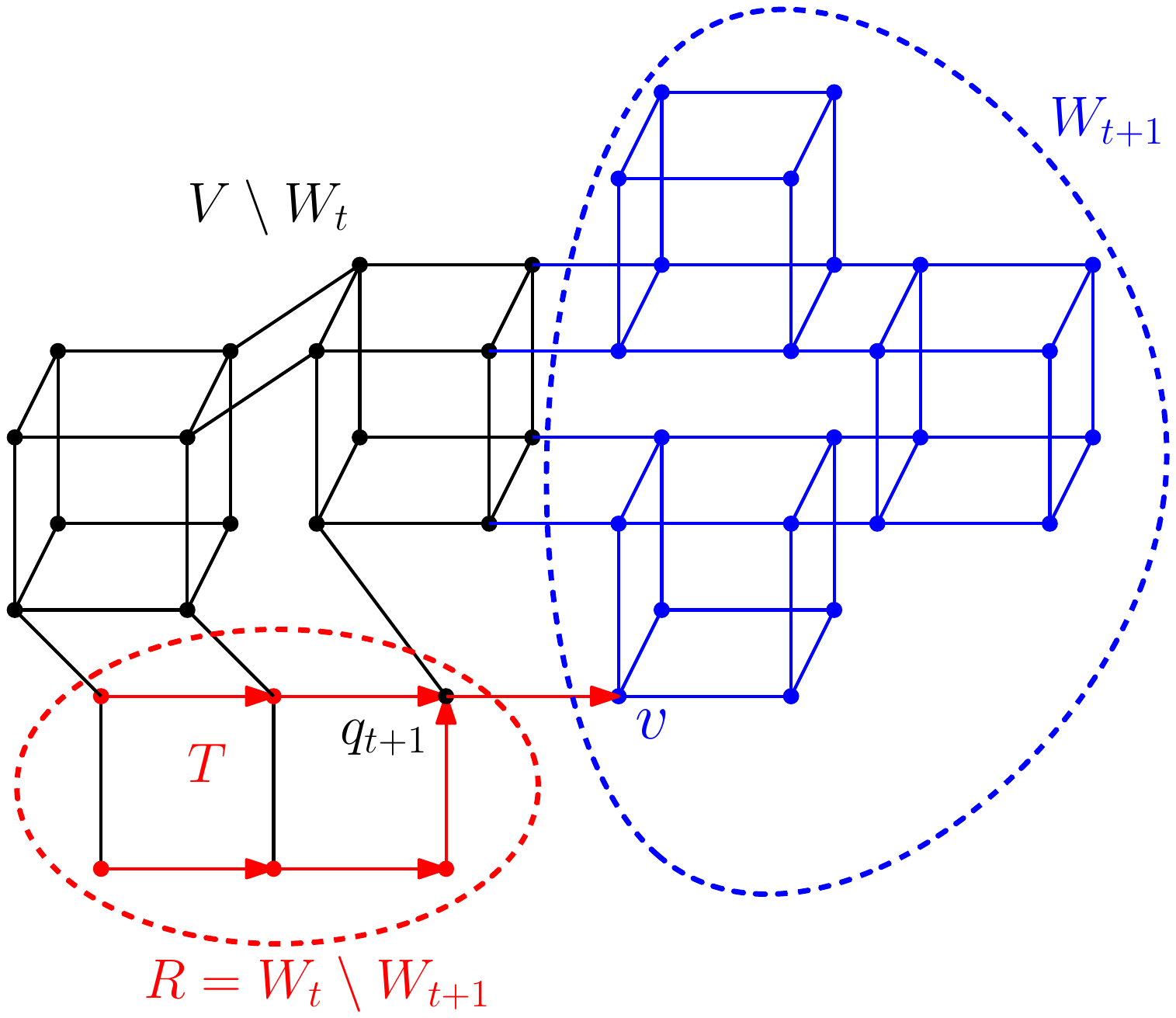}
\caption{The figure demonstrates the recursive definition of
$W_{t+1}$ from $W_t$ when the queried vertex is $q_{t+1}$. The
figure also depicts the construction of the tree $T$.
}
\end{center}
\end{figure}

Let $Q_t\subset V$ be the set of queries that the
querier asks in first $t$ rounds. We define $W_t\subset V$ to be the largest connected
component of $V\setminus Q_t$. Note that, either $W_{t+1}\subseteq W_t$ or
$W_{t+1}\cap W_t=\emptyset$. But, for any $t < \frac{2^n}{n^2+1}$, the second relation cannot hold, because $|W_{t+1}|,|W_t|> 2^{n-1}$ by Lemma \ref{lemma:giant}.
Therefore, for all $t < \frac{2^n}{n^2+1}$, we have $W_{t+1}\subseteq W_t$.

As the interaction continues, the adversary will progressively assign difference values $d(s)$ that satisfy constraint (\ref{eq-cons}). For all $t$, the adversary maintains the following three properties for $W_t$: 

\begin{itemize}
\item[(P1)] For every $s\notin W_t$, all the components of $d(s)$ have already been assigned a value, i.e., vector $d(s)$ is completely defined, and $D(s)=-1$. 

\item[(P2)] For every edge $(s,s\lnot i)$ such that $s,s\lnot i \in W_t$ the difference values $d_i(s)$ and $d_i(s\lnot i)$ are unassigned, in other words, these values have not been set so far.

\item[(P3)] Every vertex $s\notin W_t$ that has an edge into $W_t$ satisfies $s\in Q_t$ (this property will implicitly hold since $W_t$ is an entire connected component in $V \setminus Q_t$).
\end{itemize}

At time $t+1$ the adversary assigns vectors $d(s)$ for every vertex
$s\in W_{t}\setminus W_{t+1}$ as follows.

Denote $R= W_{t}\setminus W_{t+1}$. $R$ is the set of vertices that
got disconnected from $W_t$ because some vertex, say $q_{t+1}$, got queried.  That is, removing $q_{t+1}$ disconnected $W_t$ (see Figure \ref{fig:proof}). In general, $R$ might contain other vertices besides  
$q_{t+1}$. Nonetheless, $R$ has to be a connected
set of vertices. Let $v \in W_{t+1}$ be a neighbor of $q_{t+1}$---if $W_t$ gets disconnected by removing $q_{t+1}$ then such a vertex must exist. By construction, the vertex set  $R \cup \{v\}$ forms a connected component (see Figure
\ref{fig:proof}). Write $T$ to denote some spanning tree of $R\cup \{v\}$. We root the tree $T$ at $v$ and define
vectors $d(s)$ for every $s\in R$ in a bottom-up manner: from the leafs
to the root.

Given a leaf $s$ with the edge $(s,s\lnot i)$ in $T$ (i.e., $s\lnot i$ is the parent of $s$ in $T$), we know by property (P2) that so far we have not assigned a value to $d_i(s)$. Moreover, this holds for all edges out of $s$ that end in $R$. Write index set $J=\{  j \mid s\lnot j \in R\}$. By (P2) we are free to to assign values to $d_j(s)$ and $d_j(s \lnot j)$. For $j \in J \setminus \{ i \}$, we set $d_j(s)= d_j(s\lnot j) = 0$, and finally assign  
\begin{equation}\label{eq-1}
d_i(s)=-\sum_{k\neq i} d_k(s)-1, 
\end{equation}
and $d_i(s\lnot i) = -d_i(s)$. Such an assignment ensures that constraints (\ref{eq-cons}) are satisfied for all defined values. Once values for $s$ have been assigned, we remove it from consideration and recurse over $T$. 

Using this procedure, (\ref{eq-1}) guarantees that property (P1) is
maintained; in particular, $D(s) = -1$ for all $s \in R$. In addition, since we did not assign difference
$d_i(s)$ along any edge $(s, s \lnot i )$ such that $s, s\lnot i \in W_{t+1}$, property (P2) is
maintained as well. 

As stated above, property (P3) holds because no
vertex in $R\setminus \{q_{t+1}\}$ is directly connected to
$W_{t+1}$. 

Overall, say, the querier submits $q$ queries to the adversary with $q<\frac{2^n}{(n^2+1)}$. After the $q$ queries have been processed by the adversary we split $V$ into three sets:

(1) $\FC:=V\setminus W_q$. This is the set of (fully assigned) vertices, $s$,  for which all the components of the difference vector, $d(s)$, have been assigned a value. Also, by construction, for such $s$ we have the sum of differences $D(s)=-1$.

(2) $\PC:=\{s \in W_q \mid \text{ there exists a neighbor of } s \text{
in } \FC \}$. These are vertices, $s$, with a partial assignment. That is, some of the components of the difference vector $d(s)$ have been assigned a value. Such a partial assignment follows from the fact that we must 
have fixed the difference values for several neighbors of $s$
(but not for all of them, otherwise $s$ will belong to
$\FC$).

(3) $\NC:=V\setminus (\FC\cup \PC)$. These are the vertices with no assignments. In other words, for all $s\in \NC$, the entire vector $d(s)$ is unassigned.

By Lemma \ref{lemma:giant} we know that $|W_q|> 2^{n-1}$, since
$W_q$ is the largest connected component of the graph
$V\setminus Q_q$. Therefore $|\FC|<2^{n-1}$. By property (P3) we know
that all the vertices in $\PC$ are neighbors of $Q_q$, therefore
$|\PC|\leq n |Q_q|\leq \frac{n}{n^2+1}2^n$, therefore

\begin{equation}\label{ineq-NC}
|\NC|\geq \left(1-\frac{1}{2}-\frac{n}{n^2+1}\right)2^n \geq \frac{1}{3}2^n.
\end{equation}
The last inequality holds for all $n\geq 6$.

Denote by $x$ the probability distribution over $V$ that the
querier submits after the $q$ queries. For subset $S \subset V$, write $x(S)$ to denote the cumulative probability of vertices in $S$, $x(S):= \sum_{s\in S} x(s)$. 

The adversary will complete the remaining utilities in the game using the following rules. We will show that these rules guarantee that $x$ is not a correlated
equilibrium. In other words, the adversary can always ensure that output generated by the querier is not a correlated equilibrium. 

Case 1: $x(\NC)\leq 1/6$. In such a case, by inequality
(\ref{ineq-NC}), there exists an action $v\in \NC$ such that $x(v)\leq
\frac{1}{2}2^{-n}$. The adversary constructs a spanning tree that
contains all the vertices $W_q$ with the root $v$. Now, using the
same bottom-up assignment procedure that we described before, the adversary sets the
difference values for $s \in W_q$ such that $D(s)=-1$ for all $s\neq v$. Note that this implies that $D(v)=2^n-1$, since $\sum_{s\in V} D(s)=0$. Now we have
\begin{equation}
\E_{s\sim x}[ D(s)]\leq \frac{1}{2}2^{-n}D(v)+\left(1-\frac{1}{2}2^{-n}\right)(-1)< -\frac{1}{2}+\frac{1}{2}2^{-n} < 0.
\end{equation}
By Remark~\ref{rem} we get that $x$ is not a correlated equilibrium.

Case 2: $x(\NC)> 1/6$. Denote by $M$ the maximal (difference) value that was
assigned by the adversary during the $q$ queries. Recall that
for strategy profiles $s \in \NC$ difference values $d_i(s)$, for all $i$, are unassigned. We leverage this freedom and set these difference values such that some player, say $p$, has an incentive to deviate to one of the strategies $0$ or $1$. This will imply that $x$ is not a correlated equilibrium. 

For $j=0,1$ we denote by $\NC_j$ the set actions in $\NC$ where
player $p$ plays action $j$, $N_j:=\{ s \in N \mid s_p = j\}$. One of the sets $\NC_0,\NC_1$ satisfies
$x(\NC_j)\geq \frac{1}{12}$. Without loss of generality, say $x(\NC_0)\geq \frac{1}{12}$. The
adversary sets $d_p(s)=-12M$ for every $s \in \NC_0$ and then, following equation (\ref{eq-cons}), sets $d_p$ for $s \lnot p \in \NC_1$. All the remaining unassigned difference values are set to zero. 

Now, for the action $s_p=0$
we have
\begin{equation}
\sum_{s_{-p}}  \left[ d_p(s_p,s_{-p}) \right] x(s_p, s_{-p}) \leq \frac{11}{12}M+\frac{1}{12}(-12M)<0
\end{equation}
This implies that Definition~\ref{defn:ce} does not hold for the proposed distribution $x$, and hence it is not a correlated equilibrium.

\end{proof}

\section{Discussion}\label{sec:dis}

\subsection{The Tightness of the Result by Jiang and Leyton-Brown \cite{JL}} 

Using the ellipsoid-against-hope algorithm (see also
\cite{PR}) Jiang and Leyton-Brown reduced the problem of computing correlated equilibrium in $n$-player games for each player into the following problem. 

\emph{Reduced problem\footnote{For ease of exposition, we state the problem only for binary-action games. The formulation in~\cite{JL} is for $m$-action games, with $m\geq 2$, and can be obtained by appropriately increasing the dimension of the involved vectors.} :} We are given a game with $n$ players, $2$ actions per player, and a vector $y \in \mathbb{R}_{+}^{n}$. The goal is to determine a pure action profile $s$ such that $y^T d(s)> 0 $, where $d(s)$ is the vector of differences at $s$.  


In the proof of Theorem \ref{theo:main} we show that in the pure-action-query model, if an algorithm, after submitting a sub-exponential number of queries, is unable to determine a pure-action profiles, $s$, for which $\sum_i d_i(s) > 0$ then the algorithm can not produce a correlated equilibrium. The sum, $\sum_i d_i(s)$, corresponds to the reduced problem with $y=(1,1,...,1)$. In this sense, solving the reduced problem is necessary for determining a correlated equilibrium. Note that, by scaling the utilities, we can extend the same arguments to prove the necessity of the reduced problem for \emph{any} vector $y \in \mathbb{R}_{+}^{n}$. 


Overall, in the pure actions query model, the problem of computing correlated equilibrium is \emph{equivalent} to the reduced
problem. The result of Jiang and Leyton-Brown~\cite{JL} establishes that in order to compute a correlated equilibrium in polynomial time it is sufficient to solve the reduced problem in polynomial time. On the other hand, if there exists a vector $y$ for which the reduced problem cannot be solved efficiently, then the arguments presented in this paper prove that a correlated equilibrium cannot be efficiently determined. 

\subsection{The Result of Hart and Nisan~\cite{HN}} \label{sec:dis-hn}

We show that in the pure-action-query model, there is no polynomial time \emph{deterministic} algorithm for computing
\emph{exact} correlated equilibrium. Whereas, regret-minimizing dynamics give us a 
polynomial-time \emph{randomized} algorithm for computing \emph{approximate} correlated equilibrium. In order to complete the picture one should answer the following two questions in this pure-action-query model:
(i) Does there exist a polynomial time \emph{deterministic} algorithm for computing \emph{approximate} correlated equilibrium?; (ii) Can a polynomial-time \emph{randomized} algorithm compute an \emph{exact} correlated equilibrium?

Hart and Nisan \cite{HN} show that the answer to both of these questions is negative. 

\bibliographystyle{plain}
\bibliography{qc-corr-eq}

\begin{thebibliography}{10}

\bibitem{aumann74}
Robert~J Aumann.
\newblock Subjectivity and correlation in randomized strategies.
\newblock {\em Journal of Mathematical Economics}, 1(1):67--96, 1974.

\bibitem{aumann87}
Robert~J Aumann.
\newblock Correlated equilibrium as an expression of bayesian rationality.
\newblock {\em Econometrica: Journal of the Econometric Society}, pages 1--18,
  1987.

\bibitem{chen}
Xi~Chen, Xiaotie Deng, and Shang-Hua Teng.
\newblock Settling the complexity of computing two-player nash equilibria.
\newblock {\em Journal of the ACM (JACM)}, 56(3):14, 2009.

\bibitem{Ch}
Fan~RK Chung.
\newblock {\em Spectral Graph Theory}.
\newblock {A}merican {M}athematical {S}ociety, 1997.

\bibitem{costis}
Constantinos Daskalakis, Paul~W. Goldberg, and Christos~H. Papadimitriou.
\newblock The complexity of computing a nash equilibrium.
\newblock {\em SIAM J. Comput.}, 39(1), 2009.

\bibitem{GS}
John Fearnley, Martin Gairing, Paul Goldberg, and Rahul Savani.
\newblock Learning equilibria of games via payoff queries.
\newblock In {\em Proceedings of the fourteenth ACM conference on Electronic
  commerce}, EC '13, pages 397--414, 2013.

\bibitem{FV}
Dean~P Foster and Rakesh~V Vohra.
\newblock Asymptotic calibration.
\newblock {\em Biometrika}, 85(2):379--390, 1998.

\bibitem{germano2007existence}
Fabrizio Germano and G{\'a}bor Lugosi.
\newblock Existence of sparsely supported correlated equilibria.
\newblock {\em Economic Theory}, 32(3):575--578, 2007.

\bibitem{goldberg}
Paul~W Goldberg and Christos~H Papadimitriou.
\newblock Reducibility among equilibrium problems.
\newblock In {\em Proceedings of the thirty-eighth annual ACM symposium on
  Theory of computing}, pages 61--70. ACM, 2006.

\bibitem{HM}
Sergiu Hart.
\newblock Adaptive heuristics.
\newblock {\em Econometrica}, 73(5):1401--1430, 2005.

\bibitem{open}
Sergiu Hart.
\newblock On a panel discussing future directions in algorithmic game theory.
\newblock \url{http://www.youtube.com/watch?v=tq2EvhOl2k0}, {H}ebrew
  {U}niversity, 2011.

\bibitem{hart2000simple}
Sergiu Hart and Andreu Mas-Colell.
\newblock A simple adaptive procedure leading to correlated equilibrium.
\newblock {\em Econometrica}, 68(5):1127--1150, 2000.

\bibitem{HN}
Sergiu Hart and Noam Nisan.
\newblock The query complexity of correlated equilibria.
\newblock {\em arXiv preprint arXiv:1305.4874}, 2013.

\bibitem{HPV}
Michael~D Hirsch, Christos~H Papadimitriou, and Stephen~A Vavasis.
\newblock Exponential lower bounds for finding brouwer fix points.
\newblock {\em Journal of Complexity}, 5(4):379--416, 1989.

\bibitem{JL}
Albert~Xin Jiang and Kevin Leyton-Brown.
\newblock Polynomial-time computation of exact correlated equilibrium in
  compact games.
\newblock {\em Games and Economic Behavior}, 2013.

\bibitem{LW}
N~Littlestone and MK~Warmuth.
\newblock The weighted majority algorithm.
\newblock {\em Information and Computation}, 108(2):212--261, 1994.

\bibitem{Nash}
John Nash.
\newblock Non-cooperative games.
\newblock {\em The Annals of Mathematics}, 54(2):286--295, 1951.

\bibitem{PR}
Christos~H Papadimitriou and Tim Roughgarden.
\newblock Computing correlated equilibria in multi-player games.
\newblock {\em Journal of the ACM (JACM)}, 55(3):14, 2008.

\bibitem{Tr}
Luca Trevisan.
\newblock Lecture notes on graph partitioning and expanders.
\newblock \url{cs.stanford.edu/people/trevisan/cs359g/lecture06.pdf}, 2011.

\bibitem{Young}
H~Peyton Young.
\newblock {\em Strategic learning and its limits}, volume 2002.
\newblock Oxford University Press on Demand, 2004.

\end{thebibliography}

\end{document}